\documentclass[11pt,a4paper]{article}

\usepackage[utf8]{inputenc}
\usepackage[english]{babel}

\usepackage{graphicx, xcolor}
\graphicspath{{figures/}}

\usepackage{fullpage}

\usepackage{latexsym}
\usepackage{url}
\usepackage{amsfonts,amsmath,amsthm}
\usepackage{dsfont}
\usepackage{ifthen}
\usepackage{fancybox}
\usepackage{stmaryrd}
\usepackage{xspace}
\usepackage{stmaryrd}
\usepackage{subfig}
\usepackage{pgf,tikz}
\usetikzlibrary{arrows}
\usepackage{wasysym}
\usepackage{savefnmark}

\usepackage{color}
\definecolor{blu3}{rgb}{.1,.0,.4}
\usepackage{hyperref}
\hypersetup{colorlinks=true, linkcolor=blu3, urlcolor=blu3, citecolor=blu3, pdfpagemode=UseNone, pdfstartview=} 

\newtheorem{theorem}{Theorem}

\newtheorem{lemma}[theorem]{Lemma}

\DeclareMathOperator{\area}{area}

\title{Finding a Largest-Area Triangle in a Terrain\\ in Near-Linear Time\thanks{A preliminary version of this work was presented at WADS 2021~\cite{CabelloDDM21}.}}

\author{Sergio Cabello\thanks{Faculty of Mathematics and Physics, 
			University of Ljubljana, Slovenia, and Institute of Mathematics,
			Physics and Mechanics, Slovenia. \texttt{sergio.cabello@fmf.uni-lj.si}.} 
		\and
		Arun Kumar Das\thanks{Advanced Computing and Microelectronics Unit, Indian Statistical Institute, India. \texttt{arund426@gmail.com} and
		\texttt{sandipdas@isical.ac.in}.}\saveFN\nbfootnoterepeated\
		\and
		Sandip Das\useFN\nbfootnoterepeated
		\and 
		Joydeep Mukherjee\thanks{Dept. of Computer Science, Ramakrishna Mission Vivekananda Educational and Research Institute, India. 
			\texttt{joydeep.m1981@gmail.com}.}
}

\begin{document}

\maketitle

\begin{abstract}
    A terrain is an $x$-monotone polygon whose lower boundary is a single line segment. We present an algorithm to find in a terrain a triangle of largest area in $O(n \log n)$ time, where $n$ is the number of vertices defining the terrain. The best previous algorithm for this problem has a running time of $O(n^2)$.\\ 
	
\end{abstract}

\section{Introduction}
An \emph{inclusion problem} asks to find a geometric object inside a given polygon that is optimal with respect to a certain parameter of interest. This parameter can be the area, the perimeter or any other measure of the inner object that plays a role in the application at hand. Several variants of the inclusion problem come up depending on the parameter to optimize, the constraints imposed in the sought object, as well as the assumptions we can make about the containing polygon. For example, computing a largest-area or largest-perimeter convex polygon inside a given polygon is quite well studied \cite{CabelloCKSV17,Goodman1981,holt2006}. A significant amount of work has also been done on computing largest-area triangle inside a given polygon \cite{boyce1982,chandran1992,dobkin1979,mellissaratos1992}. In the last few years, there have been new efficient algorithms for the problems of finding a largest-area triangle~\cite{hoog2020,kallus2017}, a largest-area or a largest-perimeter rectangle~\cite{cabello2016}, and a largest-area quadrilateral \cite{rote2019} inside a given convex polygon. In this paper, we propose a deterministic $O(n \log n)$-time algorithm to find a largest-area triangle inside a given terrain, which improves the best known running time of $O(n^2)$, presented in \cite{DasDM2021}. These problems find applications in stock cutting \cite{Chang1986}, robot motion planning \cite{toth2017}, occlusion culling \cite{holt2006} and many other domains of facility location and operational research. 
In general, geometric optimization problems, such as \emph{inclusion, enclosure, packing, covering, and location-allocation}, often are considered in the area of Operations Research \cite{Altay2023,Dearing2023,Korf2010}. 

A polygon $P$ is \emph{$x$-monotone} if it has no vertical edge and each vertical line intersects $P$ in an interval, which may be empty. An $x$-monotone polygon has a unique vertex with locally minimum $x$-coordinate, that is, a vertex whose two adjacent vertices have larger $x$-coordinate; see for example~\cite[Lemma 3.4]{BergCKO08}. Similarly, it has a unique vertex with locally maximum $x$-coordinate. 
    If we split the boundary of an $x$-monotone polygon at the unique vertices with maximum and minimum $x$-coordinate, we get the \emph{upper boundary} and the \emph{lower boundary} of the polygon. Each vertical line intersects each of those boundaries at most once.
    
    A \emph{terrain} is an $x$-monotone polygon whose lower boundary is a single line segment, called the \emph{base} of the terrain. The upper boundary of the terrain connects the endpoints of the base and lies above the base: each vertical ray from the base upwards intersects the upper boundary at exactly one point. Figure \ref{fig:terrain_1} shows two examples of terrains.
    
    \begin{figure}[tb]
		\centering
		\includegraphics[width=\textwidth,page=1]{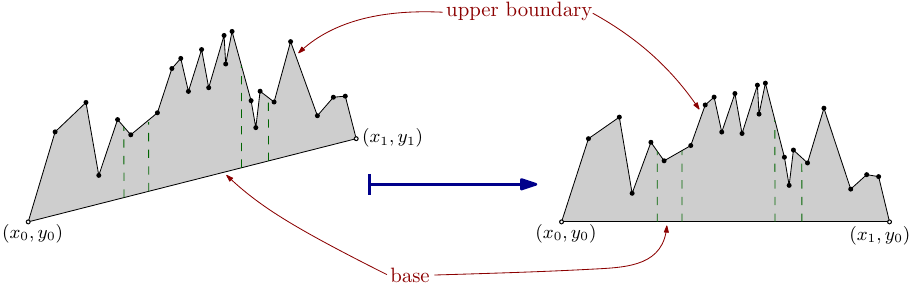}
		\caption{Two terrains. The right one is obtained from the left one by a shear transformation to make the base horizontal}
		\label{fig:terrain_1}
	\end{figure}
    
    In this work, we focus on the problem of finding a triangle of largest area inside a terrain. We will show that when the terrain has $n$ vertices, such a largest-area triangle can be computed in $O(n\log n)$ time. This is an improvement over the algorithm of Das et al.~\cite{DasDM2021}, which has a running time of $O(n^2)$. It should be noted that we compute a single triangle with largest area, even if there are more optimal solutions.
    
    To obtain our new algorithm we build on the approach and geometric insights of~\cite{DasDM2021}. More precisely, in that work, there is a single type of optimal solution that takes $O(n^2)$ time, while all the other cases can be handled in $O(n\log n)$ time. We show that the remaining case also can be solved in $O(n\log n)$ time combining shortest path trees in polygons~\cite{GuibasHLST87}, hereditary segment trees~\cite{Chazelle1994}, search for row maxima in monotone matrices~\cite{aggarwal1987}, and additional geometric insights.
    
    Our new time bound, $O(n \log n)$, is a significant improvement over the best previous result. Nevertheless, it could be that the problem is solvable in linear time. We only know that the problem cannot be solved in sublinear time because we need to scan all the vertices of the polygon. Indeed, any vertex of the terrain that is not scanned could be arbitrarily high and be the top vertex of a triangle with arbitrarily large area. Even if we assume that, say, the base and the highest vertex is also specified with the input, or the vertex that is furthest vertically from the base, it could be that the second highest vertex is the top vertex of the triangle of largest area. We leave closing this gap between $O(n\log n)$ and $\Omega(n)$ as an interesting open problem for future research. 
    
\section{Preliminaries}
    Without loss of generality, we will assume that \emph{the base of the terrain is horizontal}. The general case reduces to this one, as follows. If the endpoints of the base are $(x_0,y_0)$ and $(x_1,y_1)$ with $x_0\neq x_1$, then the shear mapping
    $(x,\,y)\mapsto \bigl( x,\, y-(x-x_0)\tfrac{y_1-y_0}{x_1-x_0}\bigr)$ 
    transforms the base to the horizontal segment connecting $(x_0,y_0)$ to $(x_1,y_0)$. Since the mapping also transforms each vertical segment to a vertical segment, the terrain gets mapped to a terrain with a horizontal base. See Figure~\ref{fig:terrain_1}. 
    Because the determinant of the Jacobian matrix of this affine transformation is $1$, the area of any measurable region of the plane is invariant under the transformation. Therefore, it suffices to find the triangle of largest area in the resulting polygon.
    
    For simplicity in the description of the algorithm, we will assume the following general position: \emph{no three vertices in the terrain are collinear}. This property is also invariant under shear transformations. The assumption can be lifted using simulation of simplicity~\cite{EdelsbrunnerM90}. More precisely, we can assume that each vertex $v_i=(x_i,y_i)$ is replaced by a vertex $v'_i=(x_i,y_i+\epsilon^i)$ for a sufficiently small $\epsilon>0$. These transformations break all collinearities if $\epsilon$ is sufficiently small. The replacement is not actually performed, but simulated. More precisely, whenever the vertices $v_i$, $v_j$ and $v_k$ are collinear, which means that
        \[
        \left| 
            \begin{matrix}
                1 & ~x_i~ & y_i\\
                1 & ~x_j~ & y_j\\
                1 & ~x_k~ & y_k
            \end{matrix} 
        \right| ~=~ 0,
        \]
the relative position of $v'_i$, $v'_j$ and $v'_k$, after the replacements, is given by the sign of the determinant
    \begin{align*}
        \left| 
            \begin{matrix}
                1 & ~x_i~ & y_i+\epsilon^i\\
                1 & ~x_j~ & y_j+\epsilon^j\\
                1 & ~x_k~ & y_k+\epsilon^k
            \end{matrix}
        \right| ~&=~ 
        \left| 
            \begin{matrix}
                1 & ~x_i~ & y_i\\
                1 & ~x_j~ & y_j\\
                1 & ~x_k~ & y_k
            \end{matrix} 
        \right| +
        \left| 
            \begin{matrix}
                1 & ~x_i~ & \epsilon^i\\
                1 & ~x_j~ & \epsilon^j\\
                1 & ~x_k~ & \epsilon^k
            \end{matrix}
        \right| ~=~ 
        \left| 
            \begin{matrix}
                1 & ~x_i~ & \epsilon^i\\
                1 & ~x_j~ & \epsilon^j\\
                1 & ~x_k~ & \epsilon^k
            \end{matrix}
        \right|\\
		&=~ \epsilon^i \cdot \left| \begin{matrix} 1 & ~x_j\\ 1 & ~x_k \end{matrix} \right|
			- \epsilon^j \cdot \left| \begin{matrix} 1 & ~x_i\\ 1 & ~x_k \end{matrix} \right| 
			+ \epsilon^k \cdot \left| \begin{matrix} 1 & ~x_i\\ 1 & ~x_j \end{matrix} \right|\\
		&= \epsilon^i (x_k-x_j) - \epsilon^j (x_k-x_i) + \epsilon^k (x_j-x_i).
	\end{align*}
    For example, if $i< j$ and $i<k$, then $\epsilon^i \gg \epsilon^j$ and $\epsilon^i \gg \epsilon^k$ whenever $\epsilon$ is positive and sufficiently small. Since the $x$-coordinates of vertices of a terrain are distinct, we get in this case, $i<\min \{j,k\}$, that the relative position of $v'_i$, $v'_j$ and $v'_k$ is given by the sign of $x_k - x_j$.
    The other cases are similar.
        
    \begin{figure}[tb]
		\centering
		\includegraphics[page=12]{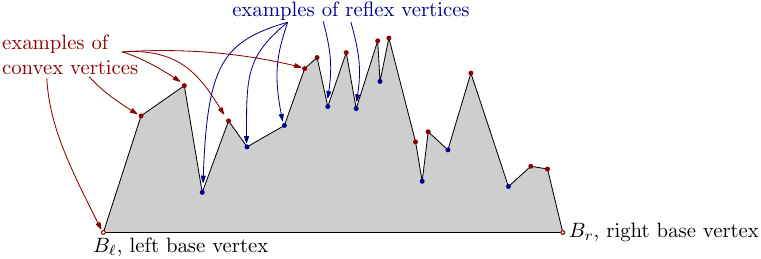}
		\caption{Notation for vertices of the terrain.}
		\label{fig:terrain_notation}
	\end{figure}

    Finally, let us introduce some notation for the vertices. See Figure~\ref{fig:terrain_notation}.
	A vertex of a terrain is \emph{convex} if the internal angle between the edges incident to this vertex is less than $\pi$ radians. If the angle is greater than $\pi$ radians, then the vertex is \emph{reflex}. Angles of $\pi$ radians do not occur because of our assumption of no $3$ collinear points. The endpoints of the base of the terrain are called \emph{base vertices}. The one with smallest $x$-coordinate is the \emph{left base vertex} and is denoted by $B_\ell$. The base vertex with largest $x$-coordinate is the \emph{right base vertex} and is denoted by $B_r$. The base vertices are convex.

\section{Previous geometric observations}
    In this section, we state several observations and properties given in~\cite{DasDM2021}, without repeating their proofs here. The first one shows that to search for an optimal solution we can restrict our attention to certain types of triangles.
    
    A triangle contained in the terrain with an edge on the base of the terrain is a \emph{grounded triangle}. For a grounded triangle, the vertex not contained in the base of the terrain is the \emph{apex}, and the edges incident to the apex are the \textit{left side} and the \textit{right side}; the right side is incident to the vertex of the base with larger $x$-coordinate. 

\begin{lemma}[Lemmas 1 and 2, Corollary 1 in~\cite{DasDM2021}]
\label{lem:structure1}
    In each terrain there is a largest area triangle $T$ satisfying all of the following properties:
    \begin{itemize}
        \item[\textup(a\textup)] the triangle $T$ is grounded;
        \item[\textup(b\textup)] the apex of $T$ lies on the boundary of the terrain \emph{or} each of the left and right sides of $T$ contains two vertices of the terrain.
    \end{itemize}
\end{lemma}

    \begin{figure}[tb]
		\centering
		\includegraphics[width=\textwidth,page=13]{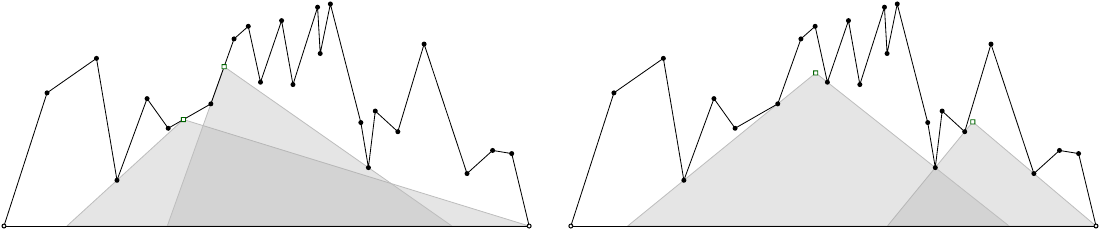}
		\caption{Examples of grounded triangles like those in Lemma~\ref{lem:structure1}. The apex is marked with a square and may lie on the boundary of the terrain (left) or not (right).}
		\label{fig:structure1}
	\end{figure}

    From now on, we restrict our attention to triangles that satisfy the properties of Lemma~\ref{lem:structure1}, and select one with largest area. Note that property (b) splits into two cases; see Figure~\ref{fig:structure1}. An option is that the apex of the grounded triangle is on the boundary of the terrain. The other option is that each of the edges incident to the apex contains two vertices of the terrain; those vertices may be base vertices. 
The first case is already solved in $O(n\log n)$ time; this is the content of the following lemma.

\begin{lemma}[Implicit in~\cite{DasDM2021}; see the paragraph before Theorem~1]
\label{lem:boundary}
    Given a terrain with $n$ vertices, we can find in $O(n\log n)$ time the grounded triangle with largest area that has its apex on the boundary of the terrain.
\end{lemma}

    The key insight to obtain Lemma~\ref{lem:boundary} is to decompose the upper boundary of the terrain into $O(n)$ pieces with the following property: for any two points $p,q$ in the same piece of the upper boundary, the largest grounded triangle with apex at $p$ and the largest grounded triangle with apex at $q$ have the same vertices of the terrain on their left sides and the same vertices of the terrain on their right sides. When the piece is not a single point, there is  precisely one vertex on each side of the triangle, and together with the apex they define uniquely a grounded triangle. One can then write a closed formula for the area of the largest grounded triangle when the apex moves along a piece, and find its maximum in $O(1)$ time per piece. We refer to~\cite{DasDM2021} for further details.

    It remains to solve the case where the apex is \emph{not} contained on the boundary of the terrain. This means that each side of the optimal triangle contains two vertices of the terrain. Because the triangle is grounded, there are two options for each of the sides: either both vertices contained in a side are reflex vertices, or one vertex is reflex and the other is a vertex of the base. 

    Recall that $B_\ell$ is the left endpoint of the base of the terrain. Consider the visibility graph of the vertices of the terrain and let $T_\ell$ be the shortest path tree from the vertex $B_\ell$ in the visibility graph. The set of edges of $T_\ell$ are denoted by $E(T_\ell)$ and we orient them away from the root, consistent with the direction that the shortest path from $B_\ell$ would follow them. We regard $T_\ell$ indistinctly as a graph and as a geometric object, that is, a set of oriented segments connecting vertices of the terrain. See Figure~\ref{fig:segment_1} for an example.

    \begin{figure}[tb]
		\centering
		\includegraphics[width =.9\textwidth,page=2]{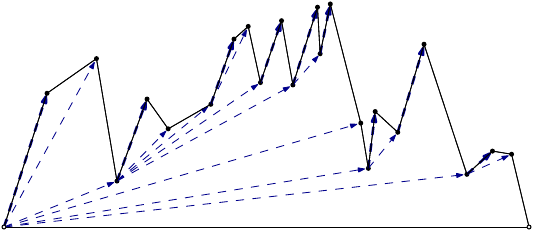}
		\caption{The tree $T_\ell$ with blue dashed arcs. Edges of $T_\ell$ contained in an edge of the polygon are thicker to make them visible.}
		\label{fig:segment_1}
	\end{figure}

    Consider an (oriented) edge $p\rightarrow q$ of $T_\ell$; the point $p$ is closer to $B_\ell$ than $q$ is. It may be that $p=B_\ell$. The \emph{forward prolongation} of $p\rightarrow q$ is the segment obtained by extending the directed segment $p\rightarrow q$ inside the terrain until it reaches the boundary of the terrain. The interior of the segment $pq$ is not part of the prolongation. The forward prolongation may be empty. However, if the forward prolongation of $p\rightarrow q$ is non-empty, then $q$ is a reflex vertex and the line supporting $p\rightarrow q$ is tangent to the boundary of the polygon at $q$. Each point on the forward prolongation is further from $B_\ell$ than $q$ is. Let $L$ be the set of non-zero-length forward prolongations of segments $p\rightarrow q$ of $T_\ell$. See Figure~\ref{fig:segment_2}. The \emph{backward prolongation} of an edge $p\rightarrow q$ of $T_\ell$ is the extension of $p\rightarrow q$ from $p$ in the direction $q\rightarrow p$ until it reaches the boundary of the terrain. The backward prolongation of $p\rightarrow q$ is empty if and only if $p=B_\ell$.

    \begin{figure}[tb]
		\centering
		\includegraphics[width =.9\textwidth,page=3]{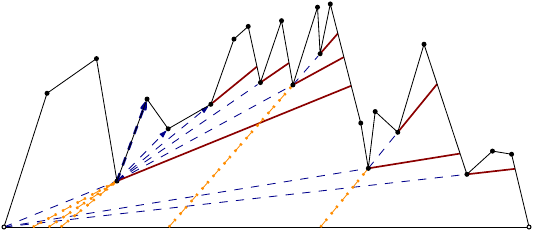}
		\caption{Forward and some backward prolongations in the example of Figure~\ref{fig:segment_1}. The set $L$, of non-empty forward prolongations of the edges of $T_\ell$, is in solid, thick red. The edges of $T_\ell$ defining $L$ are in dashed blue. In dashed-dotted orange is the set of backward prolongations for the edges defining $L$.}
		\label{fig:segment_2}
	\end{figure}
    
    A similar construction is done to obtain a shortest-path tree $T_r$ from the right endpoint $B_r$ of the base of the terrain, the prolongations of its edges, and the set $R$ of non-empty forward prolongations for the edges of $T_r$. See Figure~\ref{fig:segment_3}. Similarly, we use $E(T_r)$ to denote the set of edges of $T_r$, which may be treated as a set of oriented segments.

    \begin{figure}[tb]
		\centering
		\includegraphics[width =.9\textwidth,page=4]{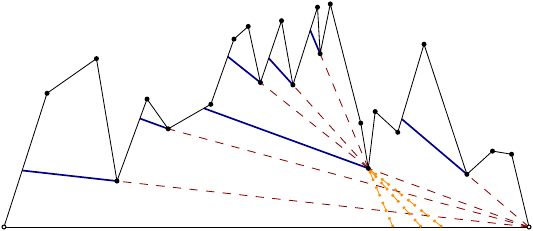}
		\caption{Forward and some backward prolongations for $T_r$ in the example of Figure~\ref{fig:segment_1}. The set $R$, of non-empty forward prolongations of the edges of $T_r$, is in solid, thick blue. The edges of $T_r$ defining $R$ are in dashed red. In dashed-dotted orange is the set of backward prolongations for the edges defining $R$.}
		\label{fig:segment_3}
	\end{figure}
    
    Using that the terrain is an $x$-monotone polygon and the lower boundary is a single segment, one obtains the following properties.
    
\begin{lemma}[Lemmas 3, 4 and 5 in~\cite{DasDM2021}]
\label{lem:apex}
    The backward prolongation of each edge in $E(T_\ell)\cup E(T_r)$ has an endpoint on the base of the terrain; it may be a vertex of the base. The segments in $E(T_\ell)$ have positive slope and the segments in $E(T_r)$ have negative slope.
    
    If the apex of the grounded triangle with largest area is \emph{not} on the boundary of the terrain, then there is an edge $s_\ell$ of $L$ and an edge $s_r$ of $R$ such that: the left side of the triangle is collinear with $s_\ell$, the right side of the triangle is collinear with $s_r$, and the apex of the triangle is the intersection $s_\ell\cap s_r$.
\end{lemma}

\section{New algorithm}
    We are now going to describe the new algorithm. In fact, we describe the missing piece in the previous approach of~\cite{DasDM2021}.
    Because of Lemma~\ref{lem:structure1}, it suffices to search for the grounded triangle of largest area. We have two cases to consider: the apex may be on the boundary of the terrain or not. The first case can be handled using Lemma~\ref{lem:boundary}. To approach the second case, we use Lemma~\ref{lem:apex}: in such a case the apex of the triangle belongs to $A=\{ s_\ell\cap s_r\mid s_\ell\in L,\, s_r\in R\}$. We refer to $A$ as the \emph{set of candidate apices}. See Figure~\ref{fig:segment_A}. Each candidate apex $s_\ell\cap s_r$ defines uniquely a triangle because its sides must be collinear with $s_\ell$, $s_r$ and the base of the terrain.

    \begin{figure}[tb]
		\centering
		\includegraphics[width =.9\textwidth,page=5]{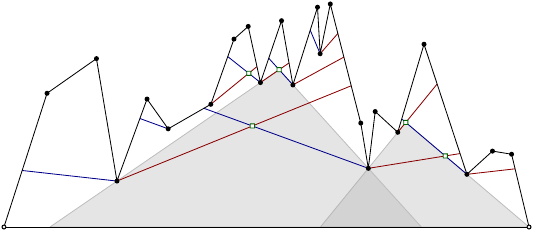}
		\caption{The candidate apices for the example of Figure~\ref{fig:segment_1} are marked with green boxes. Two of the triangles they define are shaded.}
		\label{fig:segment_A}
	\end{figure}
    
    We start providing a simple property for $L$ and $R$.
    
\begin{lemma}
\label{lem:LR}
    The edges of $L$ are pairwise interior-disjoint and can be computed in $O(n)$ time.
    The same holds for $R$.
\end{lemma}
\begin{proof}
    Consider the forward prolongation $qt\in L$ of the oriented edge $p\rightarrow q$ of $T_\ell$. The shortest path from $B_\ell$ to any point on $qt$ consists of the shortest path from $B_\ell$ to $q$ followed by a portion of $qt$. (See for example~\cite[Lemma~2.1]{GuibasHLST87} for the structure of shortest paths in a polygon without holes.)
    It follows that the edges of $L$ are contained in shortest paths from $B_\ell$ and thus they are pairwise disjoint. (They cannot overlap because of our assumption on general position.)
    
    Guibas et al. \cite{GuibasHLST87} show how to compute in $O(n)$ time the shortest path tree $T_\ell$ from $B_\ell$ and the forward extensions $L$. This is the \emph{extended algorithm} discussed after their Theorem 2.1, where they decompose the polygon into regions such that the shortest path to any point in the region goes through the same sequence of vertices of the polygon.
\end{proof}

    We use Lemma~\ref{lem:LR} to compute  $L$ and  $R$ in linear time.
    We will start using $\ell$ or $\ell_i$ for a generic segment of $L$ and $r$ or $r_j$ for a generic segment of $R$.
    Note that $L\cup R$ has $O(n)$ segments. However, the set $A$ of candidate apices may have quadratic size, as can be seen in the schematic construction of Figure~\ref{fig:segment_A2}. To get our improved running time we treat them implicitly.

    \begin{figure}[tb]
		\centering
		\includegraphics[width =.9\textwidth,page=6]{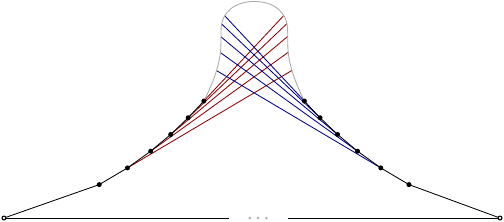}
		\caption{Example showing that we may have a quadratic number of candidate apices.}
		\label{fig:segment_A2}
	\end{figure}
	
    We use a \textit{hereditary segment tree}, introduced by Chazelle et al.~\cite{Chazelle1994}, as follows. We decompose the $x$-axis into intervals using the $x$-coordinates of the endpoints of the segments in $L\cup R$. We disregard the two unbounded intervals, that is, the leftmost and the rightmost intervals. 
    The resulting intervals are called the \emph{atomic intervals}. See Figure~\ref{fig:atom_interval} for an example where the atomic intervals are marked as $1,2,\dots,11$. We make a height-balanced binary tree $\mathcal{T}$ such that the $i$-th leaf represents the $i$-th atomic interval from left to right; see Figure \ref{fig:seg_tree}. For each node $v$ of $\mathcal{T}$, we define the interval $I_v$ as the union of all the intervals represented in the leaves of the subtree rooted at $v$.  Alternatively, for each internal node $v$, the interval $I_v$ is the union of the intervals represented by its two children. In the two-dimensional setting, the node $v$ represents the vertical strip bounded by the vertical lines passing through the end points of $I_v$. Let us denote this strip by $J_v$. In Figure \ref{fig:atom_interval}, the vertical strip $J_v$ is highlighted for the node $v$ of Figure \ref{fig:seg_tree} that is highlighted.
    
    \begin{figure}[tb]
    	\centering
    	\includegraphics[width =.8\textwidth,page=7]{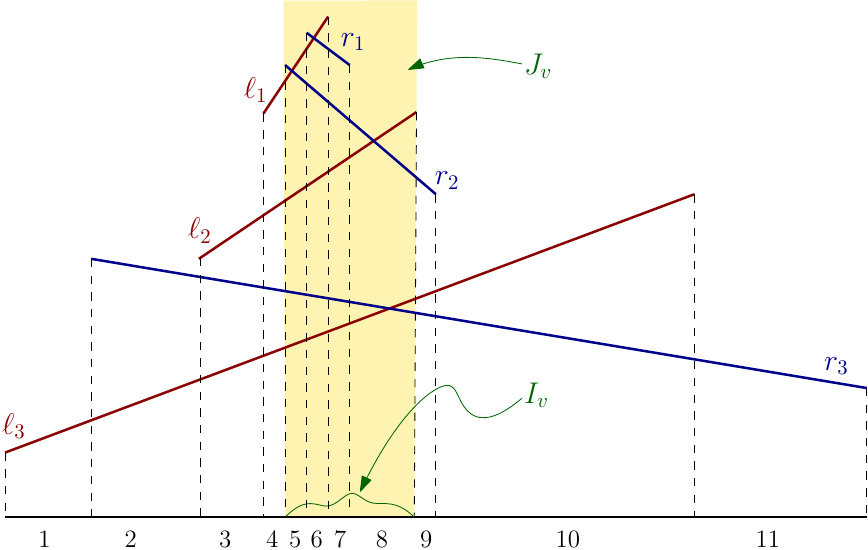}
    	\caption{Example showing the atomic intervals for a case with $|R|=|L|=3$.}
    	\label{fig:atom_interval}
	\end{figure}
	
	\begin{figure}[tb]
		\centering
		\includegraphics[width =.9\textwidth,page=8]{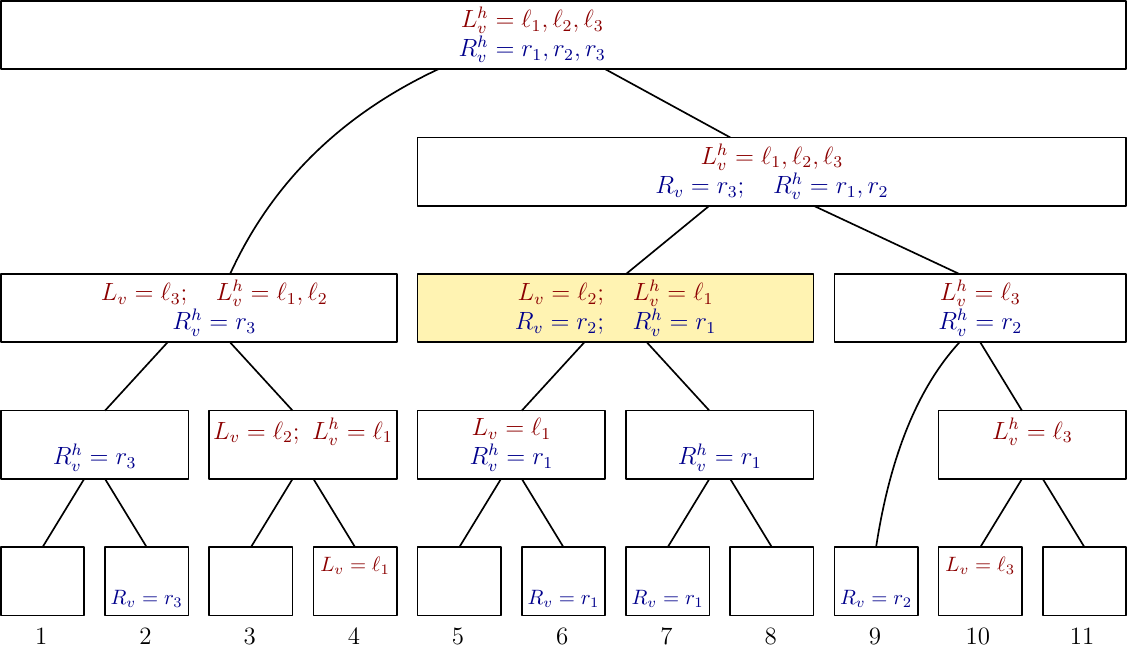}
		\caption{Hereditary segment tree for the example of Figure~\ref{fig:atom_interval}. All the lists that are not indicated are empty.}
    	\label{fig:seg_tree}
	\end{figure}
    
    Consider a node $v$ of $\mathcal{T}$ and denote by $w$ its parent.
    We maintain in $v$ four lists of segments: $L_v$, $R_v$, $L_v^h$ and $R_v^h$. The list $L_v$ contains all the segments $\ell \in L$ such that the $x$-projection of $\ell$ contains $I_v$ but does not contain $I_w$.
    Similarly, $R_v$ contains the segments $r \in R$ whose projection onto the $x$-axis contains $I_v$ but does not contain $I_w$.
    We call $L_v$ and $R_v$ the \emph{standard} lists. The list $L_v^h$ contains the members of $L_u$ for all proper descendants $u$ of $v$ in $\mathcal{T}$, that is, all descendants of $v$ excluding $v$ itself. Similarly, $R_v^h$ contains the members of $R_u$ for all proper descendants $u$ of $v$ in $\mathcal{T}$. We call $L_v^h$ and $R_v^h$ the \emph{hereditary} lists. We put only one copy of a segment in a hereditary list of a node, even if it is stored in more than one of its descendants. See Figure~\ref{fig:seg_tree} for an example. The standard lists and the hereditary lists for a node are stored explicitly at the node. 
    
    Chazelle et al.~\cite{Chazelle1994} noted that 
    \begin{equation}
        \sum_{v\text{ node of }\mathcal{T}} \bigl( |L_v|+|R_v|+|L^h_v|+|R^h_v| \bigr) ~=~ O(n\log n).
    \label{eq:boundlists}
    \end{equation}
    To obtain this bound, one first argues that each single segment $s$ of $L\cup R$ appears in $O(\log n)$ standard lists, namely in at most two nodes at each level. Moreover, the nodes that contain the segment $s$ in their standard lists have $O(\log n)$ ancestors in total, namely a subset of the nodes on the search path in $\mathcal{T}$ to the extreme atomic intervals contained in the projection of $s$. It follows that each segment $s\in L\cup R$ appears in $O(\log n)$ hereditary lists.

    For each node $v$ of $\mathcal{T}$ we define the intersections
    \begin{align*}
        A_v ~=~ &\{ \ell\cap r\mid  \ell\in L_v,\, r\in R_v,\, x(\ell\cap r)\in I_v \} \cup \\
                 &\{ \ell\cap r\mid  \ell\in L^h_v,\, r\in R_v,\, x(\ell\cap r)\in I_v \}\cup \\
                 &\{ \ell\cap r\mid  \ell\in L_v,\, r\in R^h_v,\, x(\ell\cap r)\in I_v \}.
    \end{align*}
    The set $A_v$ is the set of \emph{candidate apices defined by the node $v$}.
    Note that in the definition of $A_v$ we exclude the case when $\ell\in L^h_v$ and $r\in R^h_v$.
    
\begin{lemma}
\label{lem:union}
    The set of candidate apices, $A$, is the disjoint union of the sets $A_v$, where $v$ iterates over the nodes of $\mathcal{T}$.
\end{lemma}
\begin{proof}
	Consider a candidate apex $a\in A$. Because the segments of $L$ (and $R$) are interior disjoint by Lemma~\ref{lem:LR}, there is exactly one segment $\ell\in L$ and one segment $r\in R$ that intersect and give the apex $a=\ell\cap r$. Let $u$ be the leaf of $\mathcal{T}$ such that the $x$-coordinate of $a=\ell\cap r$ is contained in $I_u$. Let $\pi$ be the path in $\mathcal{T}$ from $u$ to the root of~$\mathcal{T}$.
	
	We walk the path $\pi$ from $u$ upwards until the first node $v$ with the property that $\ell\in L_v\cup L^h_v$ and $r\in R_v\cup R^h_v$ is reached. Since at the root $r$ we have $L^h_r=L$ and $R^h_r=R$, such a node $v$ exists. We next argue that it cannot be that $\ell\in L^h_v$ and $r\in R^h_v$. To see this, we first note that, if $v$ is a leaf, their hereditary lists are empty. Therefore, if $(\ell,r)\in L^h_v\times R^h_v$, then $v$ is not a leaf. Consider the child $w$ of $v$ in the path $\pi$. Since the interval $I_w$ of such a child $w$ contains $x(\ell\cap r)$, the interval $I_w$ intersects $x(\ell)$. Moreover, because $\ell\in L^h_v$, we conclude that $\ell\in L_w\cup L^h_w$. Similarly, $r\in R_w\cup R^h_w$, and thus $v$ was not the lowest node on $\pi$ with the desired property. We conclude that, indeed, $(\ell,r)\notin L^h_v\times R^h_v$. Finally, we note that the intersection point $a=\ell\cap r$ has its $x$-coordinate in $I_v$ because $x(\ell\cap r)\in I_u\subseteq I_v$. It follows that $a=\ell\cap r\in A_v$.
    
    To argue that the union is disjoint, we first note that only nodes $w$ along $\pi$ satisfy that $x(\ell\cap r)\in I_w$. Therefore, only for the nodes $w$ of $\pi$ can $A_w$ contain $\ell\cap r$. For each proper descendant $w$ of $v$ along $\pi$, we have $\ell\notin L_w$ or $r\notin R_w$ because of the definition of $v$ as the lowest node of $\pi$ satisfying $\ell\in L_v\cup L^h_v$ and $r\in R_v\cup R^h_v$. Therefore, $a=\ell\cap r$ does not belong to $A_w$ for any descendant $w$ of $v$ along $\pi$. For each ancestor $w$ of $v$ along $\pi$, we will have $\ell\in L^h_w$ and $r\in R^h_w$, and therefore because of the definition of $A_w$ we have $a\notin A_w$ for those ancestors.
\end{proof}

    We have to find the best apex in $A$. Since $A=\bigcup_v A_v$ because of Lemma~\ref{lem:union}, it suffices to find the best apex in $A_v$ for each $v$. (We do not use the property that the union is disjoint.) For this we consider each $v$ separately and look at the interaction between the lists $L_v$ and $R_v$, the lists $L_v^h$ and $R_v$, and the lists $R_v^h$ and $L_v$. 
    
\subsection{Interaction between two standard lists}
\label{sec:stst}

   Consider a fixed node $v$ and its standard lists $L_v$ and $R_v$. The $x$-projection of each segment in $L_v\cup R_v$ is a superset of the interval $I_v$, and thus no endpoint of such a segment lies in the interior of $J_v$. See Figure~\ref{fig:standard} to follow the discussion in this section.

\begin{figure}[tb]
		\centering
		\includegraphics[scale=1.1,page=9]{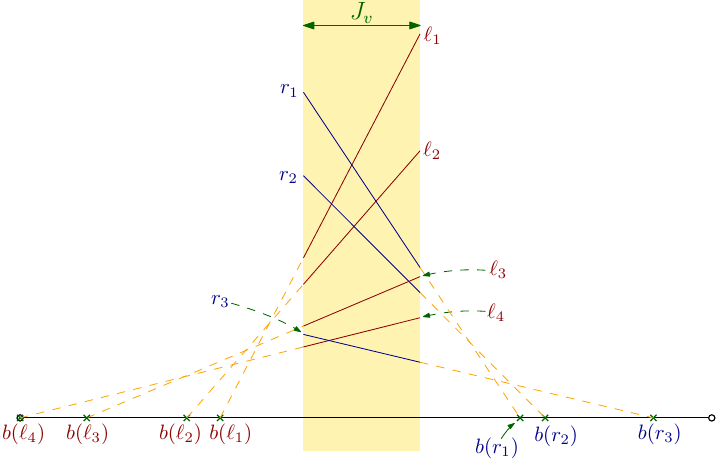}
		\caption{Interaction between two standard lists. The segments in solid may be longer, but the part contained in $J_v$ is contained in the segment. The dashed portion may contain part of the segment and contains the whole backward prolongation of the segment.}
    	\label{fig:standard}
\end{figure}
    
    Since the segments in $L_v$ are pairwise interior-disjoint (Lemma~\ref{lem:LR}) and they cross the vertical strip $J_v$ from left to right, we can sort them with respect to their $y$-order within the vertical strip $J_v$. We sort them in decreasing $y$-order. Henceforth, we regard $L_v$ as a sorted list. Thus, $L_v$ contains $\ell_1,\dots,\ell_{|L_v|}$ and, whenever $1\le i< j\le |L_v|$, the segment $\ell_i$ is above $\ell_j$.
    We do the same for $R_v$, also by decreasing $y$-coordinate. Thus, $R_v$ is a sorted list $r_1,\dots,r_{|R_v|}$ and, whenever $1\le i< j\le |R_v|$, the segment $r_i$ is above $r_j$. 

	We will discuss in Section~\ref{sec:together} how the sorted lists $L_v$ and $R_v$ for all nodes $v$ together can be obtained in $O(|L|\log |L| + |R|\log |R|)=O(n \log n)$ time. For the time being, we assume that $L_v$ and $R_v$ are already sorted.

    Because of Lemma~\ref{lem:apex}, each segment $s$ of $L\cup R$ can be prolonged inside the terrain until it hits the base of the terrain. Indeed, such a prolongation contains an edge of $E(T_\ell)\cup E(T_r)$ by definition. Let $b(s)$ be the point where the prolongation of $s$ intersects the base of the terrain. 
    
    \begin{lemma}
    \label{lem:strip_base_height}
        If $1\le i<j\le |L_v|$, then $b(\ell_i)$ lies to the right of $b(\ell_j)$.
        If $1\le i<j\le |R_v|$, then $b(r_i)$ lies to the left of $b(r_j)$.
    \end{lemma}
    \begin{proof}
        Let $s_i$ be the longest segment that contains $\ell_i$ and is contained in the terrain; let $s_j$ be the longest segment that contains $\ell_j$ and is contained in the terrain. Thus $b(\ell_i)$ is an endpoint of $s_i$ and $b(\ell_j)$ is an endpoint of $s_j$.
        Assume, for the sake of reaching a contradiction, that $b(\ell_i)$ lies to the left of $b(\ell_j)$.  This means $s_i$ and $s_j$ do not intersect to the the left of $J_v$, and thus $s_i$ is completely above $s_j$ for any $x$-coordinate to the left of $I_v$. Then $s_j$ cannot go through any vertex of the terrain to the left of $J_v$, as such a vertex would be below $s_i$, which is contained in the terrain. By construction of the hereditary segment tree, $I_v$ is a proper subset of the $x$-projections of $\ell_j$ and none of the endpoints of $\ell_j$ belongs to the interior of $J_v$. This means that the left endpoint of the segment $\ell_j$ should be to the left of $J_v$. Since the left endpoint of $\ell_j$ has to be a vertex of the terrain, we arrive at the contradiction.

        The argument for segments of $R$ is similar.
    \end{proof}

    Once $L_v$ and $R_v$ are sorted, we can detect in $O(|L_v|+|R_v|)$ time which segments of $L_v$ do not cross any segment of $R_v$ inside $J_v$. Indeed, we can merge the lists to obtain the order $\pi_\ell$ of $L_v\cup R_v$ along the left boundary of $J_v$ and the order $\pi_r$ along the right boundary of $J_v$. Then we note that $\ell_i$ does not intersect any segment of $R_v$ inside $J_v$ if and only if the rank of $\ell_i$ is the same in $\pi_\ell$ and in $\pi_r$.
    We remove from $L_v$ the segments that do not cross any segment of $R_v$ inside $J_v$. To avoid introducing additional notation, we keep denoting to the resulting list as $L_v$.
    
    Within the same running time $O(|L_v|+|R_v|)$ we can find for each $\ell_i\in L_v$ an index $\psi(i)$ such that $\ell_i$ and $r_{\psi(i)}$ intersect inside $J_v$. Indeed, if $\ell_i$ crosses some segment of $R_v$ inside $J_v$, then because $\ell_i$ has positive slope and the segments of $R_v$ have negative slope, $\ell_i$ must cross the segment of $R_v$ that in the order $\pi_\ell$ is above $\ell_i$. These segments for all $\ell_i$ can be computed with a scan of the order $\pi_\ell$.
    In the example of Figure~\ref{fig:standard}, we would have $\psi(1)=\psi(2)=\psi(3)=2$ and $\psi(4)=3$.

    Consider the $|L_v| \times |R_v|$ matrix $M= (M[i,j])_{i,j}$ defined as follows. If $\ell_i$ and $r_j$ intersect in $J_v$, then $M[i,j]$ is the area of the grounded triangle with apex $\ell_i\cap r_j$ and sides containing $\ell_i$ and $r_j$. If $\ell_i$ and $r_j$ do not intersect in $J_v$, and $j< \psi(i)$, then $M[i,j]= j\varepsilon$ for an infinitesimal $\varepsilon>0$. In the remaining case, when $\ell_i$ and $r_j$ do not intersect in $J_v$ but $\psi(i)<j$, then $M[i,j]=-j\varepsilon$  for the same infinitesimal $\varepsilon>0$. The value $\varepsilon>0$ can be treated symbolically and does not need to take any explicit value; it is only important that $n\varepsilon$ is positive and smaller than any area of any triangle we consider. (Recall that $\ell_i$ and $r_{\psi(i)}$ intersect, so there is no need to consider the option where $\ell_i$ and $r_{\psi(i)}$ do not intersect.)

	Note that within each row of $M$ the non-infinitesimal elements are contiguous. Indeed, whether a segment of $L_v$ and a segment of $R_v$ intersect in $J_v$ depends only on the orders $\pi_\ell$ and $\pi_r$ along the boundaries of $J_v$, which is the same as the order along the lists $L_v$ and $R_v$. It also follows that the entries of $M$ defined as areas of triangles form a staircase such that in lower rows it moves towards the right. It follows that a row of $M$, when we walk it from left to right, has small positive increasing values until it reaches values defined by the area of triangles, and then it starts taking small negative values that decrease. 
    
    The matrix $M$ is not constructed explicitly, but we work with it implicitly. Given a pair of indices $(i,j)$, we can compute $M[i,j]$ in constant time.
	In the example of Figure~\ref{fig:standard}, if we denote by $\alpha(\ell_i,r_j)$ the area of the grounded triangle with apex $\ell_i\cap r_j$ and sides containing $\ell_i$ and $r_j$, when the intersection $\ell_i\cap r_j$ exists, then we have
	\[
		M~=~ \left[ \begin{matrix}
				\alpha(\ell_1,r_1) & \alpha(\ell_1,r_2) & -3\varepsilon\\
				\alpha(\ell_2,r_1) & \alpha(\ell_2,r_2) & -3\varepsilon\\
				\varepsilon & \alpha(\ell_3,r_2) & -3\varepsilon\\
				\varepsilon & 2\varepsilon & \alpha(\ell_4,r_3)
			\end{matrix}
		\right] \, .
	\]

	We will show that the matrix $M$ is \emph{totally monotone}. We start with the following special case, which is the one involving more geometry.

    \begin{lemma}
    \label{lem:monotone_result_crossing}
        Consider indices $i,i',j,j'$ such that $1\le i< i'\le |L_v|$, $1\le j< j'\le |R_v|$ and such that 
		the four intersections $\ell_i\cap r_j$, $\ell_{i'}\cap r_j$, $\ell_i\cap r_{j'}$, $\ell_{i'}\cap r_{j'}$ occur
		inside $J_v$.
		If $M[i',j]> M[i',j']$, then $M[i,j]> M[i,j']$.
    \end{lemma}
    \begin{proof}
        See Figure~\ref{fig:monotone} to follow the proof.
        Because of Lemma~\ref{lem:strip_base_height}, the extensions of $\ell_i$ and $\ell_{i'}$ inside the terrain intersect in a point to the left of $J_v$, which we denote by $p_\ell$. Similarly, the extensions of $r_j$ and $r_{j'}$ intersect in a point $p_r$ to the right of $J_v$.

        \begin{figure}
        \centering
            \includegraphics[scale=1,page=10]{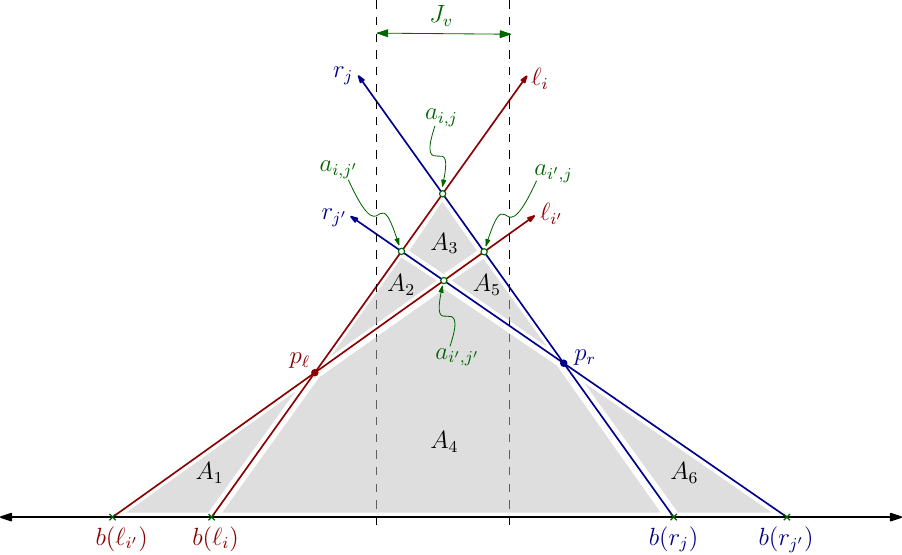}
            \caption{Scenario in the proof of Lemma~\ref{lem:monotone_result_crossing}.}
            \label{fig:monotone}
        \end{figure}

        For each $(\alpha,\beta) \in \{ i,i'\}\times \{ j,j'\}$, let $a_{\alpha,\beta}$ be the intersection point of $\ell_\alpha$ and $r_\beta$. 
        Thus, we have defined four points, namely $a_{i,j}, a_{i,j'}, a_{i',j}, a_{i',j'}$, and by assumption they are inside $J_v$.
        
        We define the following areas
        \begin{alignat*}{3}
            A_1 &= \area\bigl(\triangle (b(\ell_{i'}), b(\ell_i), p_\ell, )\bigr),~~~~~~~&&
            A_2 = \area\bigl(\triangle (p_\ell, a_{i',j'},a_{i,j'})\bigr) \\
            A_3 &= \area\bigl(\Diamond  (a_{i,j'},a_{i',j'},a_{i',j},a_{i,j})\bigr), &&
            A_4 = \area\bigl(\pentagon (b(\ell_i), b(r_j), p_r, a_{i',j'}, p_\ell) \bigr)\\
            A_5 &= \area\bigl(\triangle (a_{i',j'},p_r, a_{i',j})\bigr),&&
            A_6 = \area\bigl(\triangle (b(r_j), b(r_{j'}), p_r) \bigr).
        \end{alignat*}
        
        The condition $M[i',j]> M[i',j']$ translates into
        \[ 
            A_1 + A_4 + A_5 ~=~ M[i',j] ~>~ M[i',j'] ~=~ A_1 + A_4 + A_6, 
        \]
        which implies that $A_5 > A_6$. 
        We then have
        \[
            M[i,j] ~=~ A_2+A_3+A_4+A_5 ~>~ A_2+A_3+A_4+A_6 ~>~ A_2+ A_4+A_6~=~ M[i,j'],
        \]
        as we wanted to show.
    \end{proof}
    	
    We now consider the general case to show that the matrix $M$ is \emph{totally monotone}. In fact, the lemma restates the definition of totally monotone matrix.
    
    \begin{lemma}
    \label{lem:monotone_result}
        Consider indices $i,i',j,j'$ such that $1\le i< i'\le |L_v|$ and $1\le j< j'\le |R_v|$. If $M[i',j]> M[i',j']$, then $M[i,j]> M[i,j']$.
    \end{lemma}
    \begin{proof}
        If $\ell_{i'}$ and $r_j$ do not intersect inside $J_v$, then $M[i',j]> M[i',j']$ can only occur when $\ell_{i'}$ and $r_{j'}$ do not intersect inside $J_v$ and $\psi(i')<j'$. In this case, within the strip $J_v$ we have $\ell_i$ above $\ell_{i'}$ and above $r_{j'}$. This means that $\ell_i$ cannot intersect $r_{j'}$ inside $J_v$ and we must have also $\psi(i)<j'$. We conclude that $M[i,j]> M[i,j']$.

        If $\ell_i$ and $r_{j'}$ do not intersect inside $J_v$, we use the contrapositive. Assuming that $M[i,j]\le M[i,j']$, then because of the order within the $i$th row it must be that $\ell_i$ does not intersect $r_j$ nor $r_{j'}$ inside $J_v$ and $j<j'<\psi(i)$. That is, in this case, within $J_v$ the segment $r_j$ is above $r_{j'}$, which is above $\ell_i$. Since $\ell_{i'}$ is within $J_v$ below $\ell_i$, we then have $M[i',j']= M[i,j']\ge M[i,j] = M[i',j]$.

        It remains to consider the case where $\ell_{i'}$ and $r_j$ intersect inside $J_v$ and also $\ell_i$ and $r_{j'}$ intersect inside $J_v$. Using that $\ell_i$ is above $\ell_{i'}$, that $r_j$ is above $r_{j'}$, that $\ell_{i'}$ intersects $r_j$ inside $J_v$, and that $\ell_i$ intersects $r_{j'}$ inside $J_v$, we conclude that $\ell_i$ also intersects $r_j$ inside $J_v$ and that $\ell_{i'}$ also intersects $r_{j'}$ inside $J_v$. For this we just have to observe the relative order of the endpoints of the segments restricted to the boundaries of $J_v$. Equivalently, we get that the entries $M[i,j]$ and $M[i',j']$ are defined by areas of triangles because the staircase is moving to the right when we look at lower rows. In this case, the conclusion follows from Lemma~\ref{lem:monotone_result_crossing}.
    \end{proof}
    
    For each index $i$ with $1\le i \le |L_v|$, let $\varphi(i)$ be the smallest index of columns where the maximum in the $i$th row of $M$ is attained. Thus, $M[i,\varphi(i)]= \max \{ M[i,j]\mid 1\le j\le |R_v| \}$.
    Since $M$ is totally monotone, we can compute the values $\varphi(i)$ for all $1\le i \le |L_v|$ simultaneously using the SMAWK algorithm of Aggarwal et al.~\cite{aggarwal1987}. This step takes $O(|L_v|+|R_v|)$ time.
    
    For the node $v$ of $\mathcal{T}$, we return the maximum among the values $M[i,\varphi(i)]$ where $i=1,\ldots,|L_v|$.
    In total we have spent $O(|L_v|+|R_v|)$ time, assuming that $L_v$ and $R_v$ were already in sorted form. 
    
\subsection{Interaction between a standard list and a hereditary list}
\label{sec:sthe}
    Consider now a fixed node $v$, its standard list $L_v$ and its hereditary list $R^h_v$. The $x$-projection of each segment in $L_v$ is a superset of the interval $I_v$, and thus no endpoint of such a segment lies in the interior of $J_v$. However, the $x$-projection of a segment in $R_v^h$ has nonempty intersection with the interval $I_v$, but it is not a superset of $I_v$. This implies that each segment of $R_v^h$ has at least one of its endpoints in the interior of $J_v$. See Figure~\ref{fig:sthe} to follow the discussion in this section.
    
	Like before, we assume that the members of $L_v$ are sorted by decreasing $y$-order within $J_v$ and use $\ell_1,\ldots,\ell_{|L_v|}$ to denote them from top to bottom. See Figure~\ref{fig:sthe}. We will see in Section~\ref{sec:together} how such an order is obtained in $O(|L_v|)$ amortized time per node $v$.

\begin{figure}[tb]
		\centering
		\includegraphics[scale=1.1,page=11]{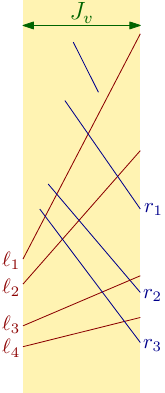}
		\caption{Interaction between the standard list $L_v$ and the hereditary list $R_v^h$. The segments in solid may be longer, but the part contained in $J_v$ is as shown.}
    	\label{fig:sthe}
\end{figure}

\begin{lemma}
\label{lem:structure}
	Each segment $r$ from the hereditary list $R_v^h$ has the following properties.
	\begin{itemize}
	\item[(a)] If $r$ intersects $\ell\in L_v$ inside $J_v$, then $r$ intersects the right boundary of $J_v$ and, at the right boundary of $J_v$, $r$ is below $\ell$.
	\item[(b)] If $r$ intersects $\ell_i$ inside $J_v$, then for each $i'\le i$ the segment $r$ intersects $\ell_{i'}$ inside $J_v$.
	\end{itemize}
\end{lemma}
	\begin{proof}
    First we show that, if an endpoint of $r\in R$ is in the strip $J_v$, then it must be above all segments $\ell\in L_v$. For this, we note that both endpoints of $r$ are on the boundary of the terrain: the rightmost endpoint is a vertex of the terrain and the leftmost endpoint is on the boundary of the terrain by definition. Since the terrain is $x$-monotone, it follows that the vertical upwards rays from both endpoints of $r$ are outside the terrain, while the segment $\ell$ is contained in the terrain.
    
    Now we show the property in (a). Assume that $r$ intersects $\ell\in L_v$ inside $J_v$. For the sake of reaching a contradiction, assume that $r$ has its right endpoint inside $J_v$. Then, because the slope of $r$ is negative, the slope of $\ell$ is positive, and the $x$-projection of $\ell$ covers $I_v$, the right endpoint of $r$ would be below $\ell$ and inside $J_v$. This contradicts the property in the previous paragraph. We conclude that $r$ intersects the right boundary of $J_v$ and, because of the slopes of $r$ and $\ell$, at the right boundary of $J_v$ the segment $\ell$ is above $r$.
    
    For property (b), it suffices to show the claim for $i'=1$: if $r$ intersects $\ell_i$ and $\ell_1$
inside $J_v$, then it intersects $\ell_{i'}$ for all $i'$ with $i'\le i$. Because of (a), $r$ intersects the right boundary of $J_v$ and, at the right boundary of $J_v$, $\ell_i$ is above $r$. It follows that, at the right boundary of $J_v$, the segment $\ell_1$ is also above $\ell_1$. Because $r$ is in the hereditary list $R^h_v$, at least one of the endpoints of $r$ is inside $J_v$, and because of property (a), this must be the left endpoint of $r$. Because of the property in the first paragraph, the left endpoint of $r$ must be above $\ell_1$. It follows that the segments $r$ and $\ell_1$ cross inside $J_v$.
    \end{proof}
   
	We proceed as follows. Because of Lemma~\ref{lem:structure}(a), we can discard any segment of $R^h_v$ that does not intersect the right boundary of $J_v$ or that is above $\ell_1$ at the right boundary of $J_v$; they do not intersect any segment of $L_v$. If no elements of $R^h_v$ remain, then no segment of $L_v$ intersects any segment of $R^h_v$ inside $J_v$ and we have finished. Otherwise, let $r_1,\ldots,r_k$ be the remaining segments of $R^h_v$, sorted from top to bottom at the right boundary of $J_v$.  We assume for the time being that the order is available; we will discuss later, in Section~\ref{sec:together}, how such an order is obtained in $O(k)$ amortized time.  We can also prune the segments of $L_v$ that do not intersect $r_k$; because of Lemma~\ref{lem:structure}(b), if they do not intersect $r_k$ inside $J_v$, they do not intersect any remaining segment of $R^h_v$ inside $J_v$. Let us keep using $L_v$ for the resulting set, to avoid introducing additional notation. Recall Figure~\ref{fig:sthe}.

	We now use an argument similar to the one for the standard lists in Section~\ref{sec:stst}. We consider a $|L_v| \times k$ matrix $M= (M[i,j])_{i,j}$ defined as follows. If $\ell_i$ and $r_j$ intersect in $J_v$, then $M[i,j]$ is the area of the grounded triangle with apex $\ell_i\cap r_j$ and sides containing $\ell_i$ and $r_j$. If $\ell_i$ and $r_j$ do not intersect in $J_v$, then $M[i,j]=j\varepsilon$ for an infinitesimal $\varepsilon>0$. This finishes the description of $M$; there is no need to consider different cases because $r_k$ crosses all remaining segments of $L_v$. At each row of $M$, we always have some areas of triangles at the right side because $r_k$ intersects inside $J_v$ all remaining segments of $L_v$. In the example of Figure~\ref{fig:sthe}, if we denote by $\alpha(\ell_i,r_j)$ the area of the triangle with apex $\ell_i\cap r_j$, when the intersection exists, then we have
	\[
		M~=~ \left[ \begin{matrix}
				\alpha(\ell_1,r_1) & \alpha(\ell_1,r_2) & \alpha(\ell_1,r_3)\\
				\alpha(\ell_2,r_1) & \alpha(\ell_2,r_2) & \alpha(\ell_2,r_3)\\
				\varepsilon & \alpha(\ell_3,r_2) & \alpha(\ell_3,r_3)\\
				\varepsilon & 2\varepsilon & \alpha(\ell_4,r_3)
			\end{matrix}
		\right] \, .
	\] 

	The same argument that was used to prove Lemma~\ref{lem:monotone_result} shows that this matrix $M$ is also totally monotone. 

    \begin{lemma}
        Consider indices $i,i',j,j'$ such that $1\le i< i'\le |L_v|$ and $1\le j< j'\le k$. If $M[i',j]> M[i',j']$, then $M[i,j]> M[i,j']$.
    \end{lemma}
	\begin{proof}
	The same arguments as in the proof of Lemma~\ref{lem:monotone_result} work. One only needs to notice that, whenever $M[i',j]$ is defined by an area, because $\ell_{i'}$ and $r_j$ intersect, then Lemma~\ref{lem:structure} implies that, for all $j'>j$, the segment $\ell_{i'}$ also intersects $r_{j'}$ inside $J_v$ and the segment $\ell_i$ intersects both $r_j$ and $r_j'$ inside $J_v$. That is, when $M[i',j]$ is defined by an area, then $M[i,j]$ is defined by an area for each $i\le i'$ and each $j'\ge j$. This implies that the crossings used in the proof of Lemma~\ref{lem:monotone_result} exist also in this setting, and thus we can again use Lemma~\ref{lem:monotone_result_crossing}.
	\end{proof}
	
	Using again the SMAWK algorithm of Aggarwal et al.~\cite{aggarwal1987}, we can find a largest area grounded triangle with apices given by the members of the sorted lists $L_v$ and $R_v^h$ in $O(|L_v| + |R_v^h|)$ time. We can also handle the interaction between $R_v$ and $L_v^h$ in a similar fashion. This finishes the description of the interaction between a standard and a hereditary list in a node $v$.

\subsection{Putting things together}
\label{sec:together}
	At each node $v$ of $\mathcal{T}$, we handle the interactions between the standard lists $L_v,R_v$, as discussed in Section~\ref{sec:stst}, and twice the interactions between a standard list and a hereditary list ($L_v,R^h_v$ is one group, and $L^h_v,R_v$ is another group), as discussed in Section~\ref{sec:sthe}. Taking the maximum over all nodes of $\mathcal{T}$, we find an optimal triangle whose apex lies in $A$ because of Lemma~\ref{lem:union}. At each node $v$ we spend $O(|L_v| + |R_v| + |L_v^h| + |R_v^h|)$, assuming that the lists are already sorted.

    To get the lists sorted at each node, we use the same technique that Chazelle et.~al~\cite[Section 3.1]{Chazelle1994} used to improve their running time. We define the following binary relation $\prec$ on the segments in $L$: for any two segments $\ell,\ell'\in L$, we have $\ell\prec \ell'$ if and only if there is a vertical line that intersects the interior of $\ell$ and $\ell'$ and, along that vertical line, $\ell$ is immediately below $\ell'$ (there are no other segments of $L$ in between). Using a left-to-right sweep-line algorithm we can compute in $O(|L| \log |L|)$ time the relations in $\prec$. Moreover, we can extend in $O(|L|)$ time this relation $\prec$ to a total order on $L$ using a topological sort. 
    Once this total order for $L$ is computed at the root of the hereditary segment tree, it can be passed to its descendants in time proportional to the lengths of the lists. The restriction of this order to any node $v$ of $\mathcal{T}$ gives the desired order for $L_v$, $L^h_v$.
A similar computation is done for $R$ separately.

    \begin{theorem}
        \label{thm:nlogn}
        A triangle of largest area inside a terrain with $n$ vertices can be found in $O(n \log n)$ time.
    \end{theorem}
    \begin{proof}
        The computation of the total order extending the above-below relation takes $O(n \log n)$ for $L$ and for $R$. After this,
        we can pass the sorted lists to each child in time proportional to the size of the lists. Thus, we spend additional  $O(|L_v|+|R_v|+|L^h_v|+R^h_v|)$ time per node $v$ of the hereditary tree to get the sorted lists.
        
        Once the lists at each node $v$ of the hereditary tree are sorted, we spend $O(|L_v|+|R_v|+|L^h_v|+R^h_v|)$ time to handle the apices of $A_v$, as explained above. Using the bound in equation \eqref{eq:boundlists}, the total time over all nodes together is $O(n\log n)$.
    \end{proof}
    
\section{Conclusions}
	We have provided an algorithm to find a triangle of largest area contained in a terrain described by $n$ vertices in $O(n \log n)$ time. The obvious open problem is whether the problem can be solved in linear time. It would also be interesting to identify other classes of polygons where the largest-area triangle can be computed in near-linear time.

\paragraph{Funding.}
Funded in part by the Slovenian Research and Innovation Agency (P1-0297, J1-9109, J1-1693, J1-2452, N1-0218).
Funded in part by the European Union (ERC, KARST, 101071836). Views and opinions expressed are however those of the authors only and do not necessarily reflect those of the European Union or the European Research Council. Neither the European Union nor the granting authority can be held responsible for them.

\bibliographystyle{abbrv}
\bibliography{Biblio}
\end{document}